\documentclass{ws-jaa}

\usepackage{amssymb}
\usepackage{amsmath,amsfonts}
\usepackage{multirow}
\usepackage{color}
\usepackage[algo2e,ruled,lined,linesnumbered,longend]{algorithm2e}
    \SetArgSty{textsl}
    \SetCommentSty{textnormal}
    \SetKwFor{For}{for}{}{endfor}

\usepackage{hyperref}

\begin{document}

\markboth{M. Borges et al.}
{Computing coset leaders and leader codewords of binary codes}

\catchline{}{}{}{}{}
%
%%%%%%%%%%%%%%%%%%%%%%%%%%%%%%%%%%%%%%%%%%%%%%%%%%%%%%%%%%%%%%%%%%%%

\title{COMPUTING COSET LEADERS AND LEADER CODEWORDS OF BINARY CODES}

\author{M. BORGES-QUINTANA\thanks{Partially funded by  RISC-Linz  DK-Doctoral Program.} \, \and M.A.~BORGES-TRENARD}

\address{Department of Mathematics, Faculty of Mathematics and Computer
Science\\  
 Universidad de Oriente, Santiago de Cuba, Cuba\\ 
\email{{\{mijail,mborges\}@csd.uo.edu.cu}}}

\author{I.~M\'ARQUEZ-CORBELLA}
\address{GRACE Project, INRIA Saclay \& LIX, CNRS UMR 7161 - \'Ecole Polytechnique, 91120 Palaiseau Cedex, France.\\
\email{irene.marquez-corbella@inria.fr}}

\author{E.~MART\'INEZ-MORO\thanks{Third and fourth authors are
partially supported by Spanish MCINN under project MTM2012-36917-C03-02.}}
\address{Institute of Mathematics IMUVa, University of Valladolid\\
Valladolid, Castilla, Spain\\
\email{edgar@maf.uva.es}}

\maketitle

\begin{history}
\received{(Day Month Year)}
\revised{(Day Month Year)}
\comby{[editor]}
\end{history}

\begin{abstract}

In this paper we use the Gr\"obner representation of a binary linear code $\mathcal C$ to give efficient algorithms for computing the whole set of coset leaders, denoted by $\mathrm{CL}(\mathcal C)$ {and the set of leader codewords, denoted by $\mathrm L(\mathcal C)$}. The first algorithm could be adapted to provide not only the Newton and the covering radius of $\mathcal C$ but also to determine the coset leader weight distribution. { Moreover, providing the set of leader codewords we have a test-set for decoding by a gradient-like decoding algorithm. Another contribution of this article is the relation stablished between zero neighbours and leader codewords.}
\end{abstract}



\section{Introduction}

The {first} goal of this article 
 is to discuss a general algorithm that produces an ordered list of the whole set of coset leaders, denoted by $\mathrm{CL}(\mathcal C)$, of a given binary code $\mathcal C$. This algorithm explains the procedure in \cite[\S 11.7]{huffman:2003} in a more transparent way and it can be adapted to determine coset leader weight distribution, the Newton radius and the covering radius which is one of the most important and studied parameters of a linear code. In \cite{raddum:2004}
some general bounds on the Newton radius for binary linear codes are given.
Finding the distribution $\left( \alpha_0, \ldots, \alpha_n\right)$ of cosets leaders (WDCL) for a code $\mathcal C$ is a classic problem in Coding Theory, see for instance \cite[Chapter 1, Section 5]{macwilliams:1977}. 
This problem is still unsolved for many family of linear codes even for first-order Reed-Muller codes (see \cite{kurshan:1972}).

{Our principal contributions are the efficient computation of the set of all coset
leaders and the definition and computation of the set of leader codewords (which is a subset of the set of zero-neighbours), by taking advantage of the additive structure of the cosets. In addition, we proved some properties of this set of leader codewords.} Note that the structure described in this paper is related  to the monotone structure of the sets of correctable
and uncorrectable errors also introduced in \cite{helleseth:2005}, where they describe the minimal uncorrectable errors under the ordering and the so-called larger halves of minimal codewords. Moreover, they use this description to give a gradient-like decoding algorithm. The same approach is considered in \cite{yasunaga:2010}.
Note that, if the decoding is done using \emph{minimum distance decoding}, a decoding failure occurs if and only if the true error is not a coset leader. Also solving the  \emph{$t$-bounded distance decoding problem} for a general linear code is related with the knowledge of the coset leaders of the code. 
Finally  the set of coset leaders in linear codes has been also related to the set of minimal support codewords which have been used in maximum likelihood decoding analysis  \cite{barg:1998, marquez:2011} and in secret sharing schemes since they describe the minimal access structure \cite{massey1993}. 

All these problems related to the one that concerns this paper are all considered to be hard computational problems (see for instance \cite{barg:1998,berlekamp:1978}) even if preprocessing is allowed \cite{bruck:1990}.

\paragraph{Outline of the paper:} 
In Section \ref{Section2} we have compiled some basic facts on coding theory and the Gr\"obner representation of binary linear codes. As for prerequisites, the reader is expected to be familiar with these topics. However we will touch only a few aspect of the theory of Gr\"obner bases since the paper is written in a ``Gr\"obner bases''-free context. For a deeper discussion of Gr\"obner representation for linear codes we refer the reader to \cite{borges:2007a} where recent results and some applications are indicated, in order to get a general picture on the subject we recommend \cite{borges:2007b}.

{In Section \ref{Section3} we provide an algorithm to compute the set of all coset leaders ($\mathrm{CL}(\mathcal C)$). A similar algorithm for computing the set of all coset leaders for a binary code follows intuitively from \cite[Chapter 11]{huffman:2003}. However, the algorithm proposed in this section do not only provide the set $\mathrm{CL}(\mathcal C)$, but also a Gr\"obner representation of $\mathcal C$ which allows the description of a complete decoding algorithm for $\mathcal C$. Moreover, this algorithm is crucial in order to derive an algorithm for the computation of the set of leader codewords.} 
The example presented at the end of this section of a binary linear code with $64$ cosets and $118$ coset leaders suggests extra applications of the algorithm such as how to obtain the weight distribution of the coset leaders or the Newton and Covering radius of a code. These applications do not pose a large additional cost to the proposed algorithm.

Section \ref{Section4} is devoted to show how the previous algorithm can be adapted to compute a test set for the code which we refer to as \emph{leader codewords}. Not only do we prove that they are zero neighbours but also that the knowledge of the set of leader codewords can be used to compute all coset leaders corresponding to a given received word.

In the final section  we point out where to find out some implementations of the algorithms presented in this paper.

\section{Preliminaries}
\label{Section2}

By $\mathbb Z$, $\mathbb K$, $\mathbb K[\mathbf X]$ and $\mathbb F_q$ we denote the ring of integers, an arbitrary finite field, the polynomial ring in $n$ variables over the field $\mathbb K$ and the finite field with $q$ elements.

A \emph{linear code} $\mathcal C$ over $\mathbb F_2$ of length $n$ and dimension $k$, or an $[n,k]$ binary code for short, is a $k$-dimensional subspace of $\mathbb F_2^n$. We will call the vectors $\mathbf v$ in $\mathbb F_2^n$ words and the particular case where $\mathbf v\in \mathcal C$, codewords. For every vector $\mathbf y\in \mathbb F_2^n$ its \emph{support} is define as its support as a vector in $\mathbb F_2^n$, i.e. 
$\mathrm{supp}(\mathbf y) = \left\{ i \mid y_i \neq 0\right\}$ and its \emph{Hamming weight}, denoted by $\mathrm{w}_H(\mathbf y)$ as the cardinality of $\mathrm{supp}(\mathbf y)$.

The \emph{Hamming distance}, $d_H(\mathbf x, \mathbf y)$, between two vectors $\mathbf x, ~\mathbf y \in \mathbb F_2^n$ is the number of places where they differ, or equivalently, 
$d_H(\mathbf x, \mathbf y) = \mathrm{w}_H(\mathbf x - \mathbf y)$. The \emph{minimum distance} $d(\mathcal C)$ of a linear code $\mathcal C$ is defined as the minimum weight among all nonzero codewords. 

Choose a parity check matrix $H$ for $\mathcal C$. The \emph{Syndrome} of a word $\mathbf y\in \mathbb F_2^n$ with respect to the parity check matrix $H$ is the vector $S(\mathbf y) = H\mathbf y^T \in \mathbb F_2^{n-k}$. As the syndrome of a codeword is $\mathbf 0$, then we have a way to test whether the vector belongs to the code. Moreover, there is a one-to-one correspondence between cosets of $\mathcal C$ and values of syndromes.

\begin{definition}
\label{CL-Definition}
The words of minimal Hamming weight in the cosets of $\mathbb F_2^n / \mathcal C$ are the set of coset leaders for $\mathcal C$ in $\mathbb F_2^n$. We will denote by $\mathrm{CL}(\mathcal C)$ the set of coset leaders of the code $\mathcal C$ and by $\mathrm{CL}(\mathbf y)$ the subset of coset leaders corresponding to the coset $\mathcal C+ \mathbf y$. We define the \emph{weight of a coset} as the smallest Hamming weight among all vectors in the coset, or equivalently the weight of one of its leaders.
\end{definition}

The zero vector is the unique coset leader of the code $\mathcal C$. Moreover, every coset of weight at most $t$ has a unique coset leader, where $t= \lfloor \frac{d(\mathcal C)-1}{2}\rfloor$ is the \emph{error-correcting capacity} of $\mathcal C$ and $\lfloor \cdot \rfloor$ denotes the greatest integer function.

For all $r\in \mathbb Z_{\geq 0}$ and $\mathbf v\in \mathbb F_2^n$ the set $\mathrm B (\mathbf v, r) : = \left\{ \mathbf w \in \mathbb F_2^n \mid d_H(\mathbf v, \mathbf w) \leq r\right\}$
is called balls around $\mathbf v$ with radius $r$ respect to the Hamming metric. Note that its cardinality is 
$|\mathrm B(\mathbf v, r)| = \sum_{i=0}^r \binom{n}{i}$.
It is well known that \emph{complete minimum distance decoding} (CDP) over the code $\mathcal C$ has a unique solution for those vectors in the union of the Hamming balls of radius $t$ around the codewords of $\mathcal C$.

From now on $\left\{ \mathbf e_i \mid i = 1, \ldots, n\right\}$ represents the canonical basis of $\mathbb F_2^n$.
The following theorem gives us a nice relationship between the coset leaders.

\begin{theorem}
\label{Theorem1}
Let $\mathbf w \in \mathrm{CL}(\mathcal C)$ such that $\mathbf w = \mathbf y + \mathbf e_i$ for some word $\mathbf y \in \mathbb F_2^n$ and $i\in \mathrm{supp}(\mathbf w)$, then $\mathbf y \in \mathrm{CL}(\mathcal C)$.
\end{theorem}

\begin{proof}
See \cite[Corollary 11.7.7]{huffman:2003}. 
\end{proof}

\begin{definition}
\label{Voronoi:Definition}
The \emph{Voronoi region} of a codeword $\mathbf c \in \mathcal C$, denoted by $\mathrm D(\mathbf c)$, is defined as:
$$\mathrm D(\mathbf c) = \left\{ \mathbf y \in \mathbb F_2^n \mid
d_H(\mathbf y, \mathbf c)\leq d_H(\mathbf y, \mathbf c') \hbox{ for all } \mathbf c'\in \mathcal C\setminus \{ \mathbf 0 \} \right\}.$$
\end{definition}

Note that the set of Voronoi regions of a binary code $\mathcal C$ covers the space $\mathbb F_2^n$. However, some points of $\mathbb F_2^n$ may be contained in several regions. Furthermore, the Voronoi region of the all-zero codeword $\mathrm D(\mathbf 0)$ coincides with the set of coset leaders of $\mathcal C$, i.e. $\mathrm D(\mathbf 0) = \mathrm{CL}(\mathcal C)$.

\begin{definition}
A \emph{test-set} $\mathcal T$ for a given binary code $\mathcal C$ is a set of codewords such that every word $\mathbf y$ either lies in the Voronoi region of the all-zero vector, $\mathrm D(\mathbf 0)$, or there exists $\mathbf t\in \mathcal T$ such that $\mathrm w_H(\mathbf y- \mathbf t)< \mathrm w_H(\mathbf y)$.
\end{definition}

We define the following characteristic crossing function:
$\begin{array}{cccc}\blacktriangle: & \mathbb F_2^s & \longrightarrow & \mathbb Z^s\end{array}$ which replace the class of $0,1$ by the same symbols regarded as integers. This map will be used with matrices and vectors acting coordinate-wise.

Let $\mathbf X$ denotes $n$ variables $x_1, \ldots, x_n$ and let $\mathbf a = (a_1, \ldots, a_n)$ be an $n$-tuple of elements of the field $\mathbb F_2$. We will adopt the following notation:
$$\mathbf X^{\mathbf a} := x_1^{\blacktriangle a_1} \cdots x_n^{\blacktriangle a_n} \in \mathbb K[\mathbf X].$$
This relationship enable us to go back to the usual definition of terms in $\mathbb K[\mathbf X]$.

\begin{definition}
\label{GrobnerRepresentation-Definition}
A Gr\"obner representation of an $[n,k]$ binary linear code $\mathcal C$ is a pair $(\mathcal N, \phi)$ where:
\begin{itemize}
\item $\mathcal N$ is a transversal of the cosets in $\mathbb F_2^n / \mathcal C$ (i.e. one element of each coset) verifying that $\mathbf 0 \in \mathcal N$ and for each $\mathbf n\in \mathcal N\setminus \{\mathbf 0\}$ there exists an $\mathbf e_i$ with $i\in \{1, \ldots, n\}$ such that $\mathbf n = \mathbf n'+ \mathbf e_i$ with $\mathbf n'\in \mathcal N$.

\item $\begin{array}{cccc}\phi:& \mathcal N\times \{\mathbf e_i \}_{i=1}^n & \longrightarrow & \mathcal N\end{array}$ is a function called {\rm Matphi function} that maps each pair $(\mathbf n, \mathbf e_i)$ to the element of $\mathcal N$ representing the coset that contains $\mathbf n+ \mathbf e_i$.
\end{itemize}
\end{definition}

The ideal $I(\mathcal C)$ associated with a binary code $\mathcal C$ is
$$I(\mathcal C) = \left\langle 
\mathbf X^{\mathbf w_1} - \mathbf X^{\mathbf w_2} \mid
\mathbf w_1 - \mathbf w_2 \in \mathcal C
\right\rangle \subseteq \mathbb K[\mathbf X].$$
Note that $I(\mathcal C)$ is a zero-dimensional ideal since the quotient ring $R=\mathbb K[\mathbf X]/I(\mathcal C)$ is a finite dimensional vector space (i.e. $\dim_{\mathbb K}\left(R\right)<\infty$). Moreover, its dimension is equal to the number of cosets in $\mathbb F_2^n / \mathcal C$.

Therefore, the word \emph{Gr\"obner} is not casual. 
Indeed, if we consider the binomial ideal $I(\mathcal C)$ and a total degree ordering $\prec$, and we compute the reduced Gr\"obner basis $\mathcal G$ of $I(\mathcal C)$ w.r.t. $\prec$. Then we can take $\mathcal N$ as the vectors $\mathbf w$ such that $\mathbf X^{\mathbf w}$ is a standard monomial module $\mathcal G$. Moreover, the function \emph{Matphi} can be seen as the multiplication tables of the standard monomials times the variables $x_i$ modulo the ideal $I_2(\mathcal C)$. Note that the \emph{Matphi} structure is independent of the particular chosen set $\mathcal N$ of representative elements of the quotient ring $\mathbb F_2^n / \mathcal C$. See \cite{borges:2007b, borges:2007a} for a more general treatment of these concepts.

\section{Computing the set of coset leaders}
\label{Section3}

\begin{definition}
\label{Weight-Ordering}
An ordering $\prec$ on $\mathbb F_2^n$ is a \emph{weight compatible ordering} if for any vectors $\mathbf a, ~\mathbf b \in \mathbb F_2^n$ we say  $\mathbf a \prec \mathbf b$ if 
$$\begin{array}{ccc}
\mathrm{w}_H(\mathbf a) < \mathrm{w}_H(\mathbf b) & \hbox{, or if, } &
\mathrm{w}_H(\mathbf a) = \mathrm{w}_H(\mathbf b) \hbox{ and } \blacktriangle \mathbf a \prec_1 \blacktriangle \mathbf b
\end{array}$$
where $\prec_1$ is any admissible order on $\mathbb N^n$, i.e. we will require that $\prec_1$ have the following additional properties:
\begin{enumerate}
\item For any vector $\mathbf u\in \mathbb N^n \setminus \{\mathbf 0 \}$, $\mathbf 0 \prec_1 \mathbf u$ and,
\item For any vectors $\mathbf u, \mathbf v, \mathbf w \in \mathbb N^n$, if $\mathbf u \prec_1 \mathbf v$, then $\mathbf u + \mathbf w \prec_1 \mathbf v + \mathbf w$.
\end{enumerate}
\end{definition}

Note that a weight compatible ordering $\prec$ is in general not an admissible ordering on $\mathbb F_2^n$. However, a weight compatible ordering $\prec$ on $\mathbb F_2^n$ satisfies:
\begin{itemize}
\item $\prec$ is a \emph{noetherian-ordering} since every strictly decreasing sequence in $\mathbb F_2^n$ eventually terminates (due to the finiteness of the set $\mathbb F_2^n$).
\item for every pair $\mathbf a, \mathbf b\in \mathbb F_2^n$, if $\mathrm{supp}(\mathbf a)\subset \mathrm{supp}(\mathbf b)$, then $\mathbf a \prec \mathbf b$.
\end{itemize}
Moreover, for every vector $\mathbf a\in \mathbb F_2^n$ we have that $\deg\left(\mathbf X^{\mathbf a}\right) = \mathrm w_H(\mathbf a)$, that is, a weight compatible ordering on $\mathbb F_2^n$ can be viewed as a total degree ordering on $\mathbb K[\mathbf X]$.

\begin{definition}
\label{List-Definition}
We define the object $\tt List$ is an ordered set of elements in $\mathbb F_2^n$ w.r.t. a weight compatible order $\prec$ verifying the following properties:
\begin{enumerate}
\item $\mathbf 0\in \tt{List}$.
\item If $\mathbf v \in \tt{List}$ and $\mathrm{w}_H(\mathbf v) = \mathrm{w}_H\left(N(\mathbf v)\right)$ then $\left\{ \mathbf v + \mathbf e_i \mid i \notin \mathrm{supp}(\mathbf v)\right\} \subset \tt{List}$, where
$N(\mathbf v) = \min_{\prec}\left\{ \mathbf w \mid \mathbf w \in \tt{List}\cap \left( \mathcal C+\mathbf v\right)\right\}$.
\end{enumerate}
We denote by $\mathcal N$ the set of distinct $N(\mathbf v)$ with $\mathbf v\in \tt{List}$.
\end{definition}

\begin{remark}
\label{Extra-Remark}
Observe that if the second condition of Definition \ref{List-Definition} holds for $\mathbf v \in \mathbb F_2^n$ then $\mathbf v \in \mathrm{CL}(\mathcal C)$. 
In particular, when $\mathbf v$ is the first element of ${\tt List}$ that belongs to $\mathcal C+ \mathbf v$, then $N(\mathbf v) = \mathbf v$.
\end{remark}

Next theorem states that the object $\tt{List}$ includes the set of coset leaders of a given binary linear code.

\begin{theorem}
\label{Theorem2}
Let $\mathbf w\in \mathbb F_2^n$. If $\mathbf w \in \mathrm{CL}(\mathcal C)$ then $\mathbf w \in \tt{List}$.
\end{theorem}
\begin{proof}
We will proceed by induction on $\mathbb F_2^n$ with a weight compatible ordering $\prec$.

By definition, the statement is true for $\mathbf 0 \in \mathbb F_2^n$. Now for the inductive step, we assume that the desired property is true for any word $\mathbf u \in \mathrm{CL}(\mathcal C)$ smaller than an arbitrary but fixed $\mathbf w\in \mathrm{CL}(\mathcal C)\setminus \{ \mathbf 0\}$ w.r.t.  $\prec$, i.e.
$$\hbox{if } \mathbf u \in \mathrm{CL}(\mathcal C) \hbox{ and } \mathbf u \prec \mathbf w \hbox{ then }\mathbf u \in \tt{List},$$
and show that this implies that $\mathbf w \in \tt{List}$.

First note that $\mathbf w$ can be written as $\mathbf w = \mathbf v + \mathbf e_i$ with $i \in \mathrm{supp}(\mathbf w)$ and $i \notin \mathrm{supp}(\mathbf v)$, or equivalently $\mathrm{supp}(\mathbf v)\subset \mathrm{supp}(\mathbf w)$, i.e. $\mathbf v \prec \mathbf w$. Moreover, since $\mathbf w \in \mathrm{CL}(\mathcal C)$, then by Theorem \ref{Theorem1} $\mathbf v$ also belongs to $\mathrm{CL}(\mathcal C)$, thus $\mathrm w_H(\mathbf v) = \mathrm w_H\left(N(\mathbf v)\right)$. So, if we invoke the induction hypothesis we have that $\mathbf v\in \tt{List}$. We now apply property $2$ of Definition \ref{List-Definition} which gives as claimed, that $\mathbf w = \mathbf v + \mathbf e_i \in \tt{List}$.
\end{proof}

Theorem \ref{Theorem2} and its proof suggest Algorithm \ref{Algorithm::1} for computing the whole set of coset leaders of a given binary code $\mathcal C$.

%%%%% ALGORITMO para calcular el conjunto de coset-leaders
\begin{algorithm2e}[!h]
\KwData{A weight compatible ordering $\prec$ and a parity check matrix $H$ of a binary code $\mathcal C$.} 
\KwResult{The set of coset leaders $\mathrm{CL}(\mathcal C)$ and $(\mathcal N, \phi)$ a Gr\"obner representation for $\mathcal C$.}
${\tt Listing} \longleftarrow [0]$;
$\mathcal N \longleftarrow \emptyset$;
$r \longleftarrow 0$;
$\mathrm{CL}(\mathcal C) \longleftarrow \emptyset$;
$\mathcal S \longleftarrow \emptyset$\;
\While{${\tt Listing} \neq \emptyset$}
{
	$\mathbf t \longleftarrow {\tt NextTerm}[{\tt Listing}]$\;
	$\mathbf s \longleftarrow \mathbf t H^T$\;
	$j \longleftarrow {\tt Member}[s,\mathcal S]$\;
	\eIf{$j \neq {\tt false}$}
	{
		\For{$k\in \mathrm{supp}(\mathbf t) ~:~\mathbf t = \mathbf t'+\mathbf e_k$ with $\mathbf t'\in \mathcal N$}{
			$\phi (\mathbf t', \mathbf e_k) \longleftarrow \mathbf t$
		}
		\If{$\mathrm{w}_H(\mathbf t) = \mathrm{w}_H(\mathbf t_j)$}
		{
		$\mathrm{CL}(\mathcal C)[j] \longleftarrow \mathrm{CL}(\mathcal C)[j] \cup \{ \mathbf t\}$\;
		${\tt Listing} \longleftarrow {\tt InsertNext}[\mathbf t, {\tt Listing}]$\;
		}
	}
	{ 
		$r\longleftarrow r+1$;  
		$\mathbf t_r \longleftarrow \mathbf t$;
		$\mathcal N \longleftarrow \mathcal N \cup \{ \mathbf t_r\}$\;
		$\mathrm{CL}(\mathcal C)[r] \longleftarrow \{ \mathbf t_r\}$;
		$\mathcal S[r] \longleftarrow  \mathbf s$\;
		${\tt Listing} = {\tt InsertNext}[\mathbf t_r, {\tt Listing}]$\;
		\For{$k\in \mathrm{supp}(\mathbf t_r) ~:~\mathbf t_r = \mathbf t'+ \mathbf e_k$ with $\mathbf t'\in \mathcal N$}
		{
			$\phi(\mathbf t', \mathbf e_k) \longleftarrow \mathbf t_r$\;
			$\phi(\mathbf t_r, \mathbf e_k) \longleftarrow \mathbf t'$\;
		}
	}
} 
\caption{Computation of $\mathrm{CL}(\mathcal C)$}
\label{Algorithm::1}
\end{algorithm2e}
%%%%%%%%% Find del algoritmo 1

The subfunctions used in Algorithm \ref{Algorithm::1} are:
\begin{itemize}
\item ${\tt InsertNext}[\mathbf t, {\tt Listing}]$, adds to $\tt Listing$ all the sums $\mathbf t + \mathbf e_k$ with $k \notin \mathrm{supp}(\mathbf t)$, removes duplicates and keeps $\tt Listing$ in increasing order w.r.t. the ordering $\prec$.
\item ${\tt NextTerm}[{\tt Listing}]$, returns the first element from $\tt Listing$ and deletes it. If $\tt Listing$ is empty returns $\emptyset$.
\item ${\tt Member}[\hbox{obj}, G]$, returns the position $j$ of $\hbox{obj}$ in $G$ if $\hbox{obj} \in G$ and $\tt false$ otherwise.
\end{itemize}

\begin{remark}
\label{Remark1}
In Algorithm \ref{Algorithm::1}, first we perform subroutine $\mathbf t = \tt{NextTerm}[{\tt Listing}]$ where the element $\mathbf t$ is deleted from the set $\tt Listing$. Then subroutine ${\tt InsertNext}[\mathbf t, {\tt Listing}]$ is carried out which inserts in $\tt Listing$ all the elements of the form:
$$\mathbf t' = \mathbf t + \mathbf e_k \hbox{ with } k\notin \mathrm{supp}(\mathbf t)\hbox{, i.e. } \mathbf t'\succ \mathbf t.$$
Therefore all the new elements inserted in $\tt Listing$ are greater than those that have already been deleted from it with respect to $\prec$.
\end{remark}

\begin{theorem}
\label{Theorem3}
Algorithm \ref{Algorithm::1} computes the set of coset leaders of a given binary code $\mathcal C$ and its corresponding {\rm Matphi} function.
\end{theorem}

\begin{proof}
We build the set $\tt{List}$ formed by all the words inserted in the object $\tt{Listing}$ during Algorithm \ref{Algorithm::1}. Let us first prove that this new set is well defined according to Definition \ref{List-Definition}. 
By \textbf{Step 1}, $\mathbf 0 \in {\tt List}$ verifying property $1$ of Definition \ref{List-Definition}. In \textbf{Step 4} the syndrome of $\mathbf t = {\tt NextTerm}[{\tt Listing}]$ is computed, then we have two possible cases based on the outcome of \textbf{Step 5}:
\begin{enumerate}
\item If $j = {\tt false}$ then the coset $\mathcal C+{\mathbf t}$ has not yet been considered. Thus, according to Remark \ref{Extra-Remark}, we have that $N({\mathbf t}) = {\mathbf t}$ and \textbf{Step 17} guarantees property $2$ of Definition \ref{List-Definition}.
\item On the other hand, if $j \neq {\tt false}$, then the element $N({\mathbf t}) = \mathbf t_j$ has already been computed. However, if ${\mathbf t}\in \mathrm{CL}(\mathcal C)$, or equivalently, $\mathrm{w}_H({\mathbf t}) = \mathrm{w}_H(\mathbf t_j)$, then \textbf{Step 12} certified property $2$.
\end{enumerate}
Therefore Algorithm \ref{Algorithm::1} construct the object ${\tt List}$ in accordance with Definition \ref{List-Definition}.

Furthermore, on one hand \textbf{Step 11} and \textbf{ Step 16} assure the computation of the complete set of coset leaders of the given code; and on the other hand \textbf{Step 8}, \textbf{Step 19} and \textbf{Step 20} compute the \emph{Matphi} function. Note that \textbf{Step 18} is necessary since the first case above ensures that $N(\mathbf t_r) = \mathbf t_r$ so by Theorem \ref{Theorem1} 
$\mathbf t' \in \mathrm{CL}(\mathcal C)$. But $\mathbf t_r = \mathbf t' + \mathbf e_k$ with $k \in \mathrm{supp}(\mathbf t_r)$ so by Remark \ref{Remark1} $\mathbf t'\prec \mathbf t_r$ has already been considered on the algorithm.
Thus, we have actually proved that Algorithm \ref{Algorithm::1} guarantees the desired outputs.

Finally, notice that the cardinality of the set {\tt List} is bounded by $n$ times the cardinality of $\mathrm{CL}(\mathcal C)$ and \textbf{Step 12} and \textbf{Step 17} guarantee that when the complete set of coset leaders is computed no more elements are inserted in ${\tt Listing}$ while \textbf{Step 3} continues deleting elements from it. Thus after a finite number of steps the set $\tt Listing$ get empty. Consequently, \textbf{Step 2} give the end of the algorithm.
\end{proof}

\begin{remark}
Note that Algorithm \ref{Algorithm::1} returns $(\mathcal N, \phi)$ that fulfill Definition \ref{GrobnerRepresentation-Definition}, for correctness we refer the reader to \cite[Theorem 1]{borges:2007a}. Furthermore, by definition, those representative of the cosets given by $\mathcal N$ are the smallest terms in $\tt List$ w.r.t. $\prec$.   
\end{remark}

%\begin{remark}
%The collection of programs and procedures \textbf{GBLA\_LC} (\emph{Gr\"obner Basis by Linear Algebra and Linear Codes}) has been already presented in previous works (see for instance \cite{borges:2006,borges:2007a,borges:2008}).
%This framework consist of various files written in the GAP \cite{GAP4} language and included in GAP's package GUAVA 3.10.
%
%Now we have implemented the {\rm CLBC} algorithm and added to the collection \textbf{GBLA\_LC}. This function \textbf{GBLA}, has the same input of the algorithm and returns a list with three components given by the set of coset leaders with respect to any total degree compatible ordering, the function {\rm Matphi} and the error correction capacity of a given binary code.
%\end{remark}

\begin{remark}
{
Algorithm \ref{Algorithm::1} has some similarities with the approach that can be deduced from \cite[\S 11.7]{huffman:2003} for computing the whole set of coset leaders in a binary code.} { First of all, Algorithm \ref{Algorithm::1} explains the algorithm in \cite[\S 11.7]{huffman:2003} in a more transparent way. 
For more details, let us consider the partial order defined by }{
 $\mathbf x \leq \mathbf y$ if 
$\mathrm{supp}(\mathbf x)\subseteq \mathrm{supp}(\mathbf y)$ where $\mathbf x,\, \mathbf y$ are two elements in $\mathbb F_2^n$.
One can use this partial order on $\mathbb F_2^n$ to define a partial order on the set of cosets of a binary code $\mathcal C$ as follows: let $\mathbf y_1 + \mathcal C$ and $\mathbf y_2 + \mathcal C$ be two different cosets of $\mathcal C$ then
$$\mathbf y_1 + \mathcal C \leq \mathbf y_2 + \mathcal C ~\Longleftrightarrow~ \exists\, \mathbf x_1 \in \mathrm{CL}(\mathbf y_1) \hbox{ and } \exists \,\mathbf x_2 \in \mathrm{CL}(\mathbf y_2) \hbox{ such that } \mathbf x_1 \leq \mathbf x_2.$$

It is shown in \cite{huffman:2003} that we can order the different cosets of a binary code $\mathcal C$ as a tree from with root $\mathcal C$ and in each edge of the tree one unit is added to the weight of the coset leader with respect to the weight of its} { descendants.}

{Therefore, both algorithms coincide in the incremental weight order applied to provide the set $\mathrm{CL}(\mathcal C)$ but our approach has the following advantages:

\begin{enumerate}
\item Algorithm \ref{Algorithm::1} also returns the additive table $\phi$ associated to the addition of a unit vector to any coset. This tool is fundamental for dealing with decoding.
\item As it is shown in the next Remark, with similar ideas to our approach, the non-binary case could be also solved.
\item Moreover, as we will see in Section \ref{Section4}, 
our algorithm allows the computation of a test-set, which is a much more smaller structure than $\mathrm{CL}(\mathcal C)$ but which could be used to solve the same problems. Also in this paper it is proven that it is a subset of the so called set of zero-neighbors and it contains any minimal test set according to the cardinality.
\end{enumerate}
}
\end{remark}

{
\begin{remark}
%A similar algorithm to Algorithm \ref{Algorithm::1} could be implemented for the $q$-ary case.
%Let $\mathbb F_q$ be a finite field with $q=p^r$ elements and $p$ prime, we could adopt the convetion that $\mathbb F_q= \frac{\mathbb F_p[X]}{\left(f(X)\right)}$ where $f(X)$ is chosen such that $f(X)$ is an irreducible polynomial over $\mathbb F_p$ of degree $r$. Let $\beta$ be a root of $f(X)$, then an equivalent formulation of $\mathbb F_q$ is $\mathbb F_p[\beta]$, i.e. any element of $\mathbb F_q$ is represented as 
%$$a_0 + a_1 \beta + \ldots + a_{r-1}\beta^{r-1} \hbox{ with }
%a_i \in \mathbb F_p \hbox{ for } i \in \{0, \ldots, r-1\}.$$
%Let $\mathbf b=(b_1, \ldots, b_n)$ be an $n$-tuple of elements of the field $\mathbb F_q$ we will call the \emph{$p$-adic expansion} of the $i$-th component of $\mathbf b$ to be the $r$-tuple
%$$\left(a_{i,0}, a_{i,1}, \ldots, a_{i,r-1} \right) \hbox{ such that } b_i = a_{i,0} + a_{i,1} \beta + \ldots + a_{i,r-1}\beta^{r-1}$$
%Similarly, we define the \emph{generalized support} of the vector $\mathbf b\in \mathbb F_q^n$ as the support of the $(nr)$-tuple described by the $p$-adic expansion of all the elements of $\mathbf b$, i.e. 
%$$\mathrm{supp}_{\mathrm{gen}}(\mathbf b) = \mathrm{supp}\left(
%(a_{i,0}, \ldots, a_{i,r-1})_{i=1, \ldots, n}
%\right)$$
%Similar to the previous Remark, given two vectors $\mathbf x$ and $\mathbf y$ in $\mathbb F_q^n$ we define the following partial ordering:
%$$\mathbf x \leq \mathbf y ~\Longleftrightarrow~ \mathrm{supp}_{\mathrm{gen}}(\mathbf x)\subseteq \mathrm{supp}_{\mathrm{gen}}(\mathbf y)$$
Also the same idea could be implemented in the most general case of
linear codes over  {$\mathbb F^n_q$, with $q=p^r$ and $p$ a prime. If we
define for} $\mathbf x$ and $\mathbf y$ in $\mathbb F^n_q$,
define $\mathbf x\leq\mathbf y$ provided that $\mathrm{supp}(\mathbf
x_i) \subseteq \mathrm{supp}(\mathbf y_i)$ for all $i=1,\ldots, n$ where
$\mathbf x_i$ is the $p$-adic expansion of the $i$th component of
$\mathbf x$. {In this case the ideas in \cite{borges:2007a}} could be used to compute
a complete set of coset representatives with an analogous incremental
structure with respect to the generalized support (but \textbf{not} with respect the the coset weights) and its
additive table $\phi$. Take notice that most of the chosen coset representatives may not be coset leaders if the weight of the coset is greater than the error-correcting capability of the code. On those cosets where the chosen representative is not a coset leader a descendant property could be also defined to find the coset leaders. However, other results of this paper as leader codewords and coset leaders can not be
straightforward deduced from our approach to the $q$-ary case.
\end{remark}
}

\begin{example}
\label{Example::1}
Consider the $[n=10,k=4,d=4]$ binary code $\mathcal C$ defined by the following parity check matrix:
$$H_{\mathcal C} = \left( \begin{array}{cccccccccc}
1 & 0 & 0 & 0 & 1 & 0 & 0 & 0 & 0 & 0 \\
1 & 0 & 1 & 1 & 0 & 1 & 0 & 0 & 0 & 0 \\
1 & 1 & 0 & 1 & 0 & 0 & 1 & 0 & 0 & 0 \\
1 & 1 & 1 & 0 & 0 & 0 & 0 & 1 & 0 & 0 \\
1 & 1 & 1 & 1 & 0 & 0 & 0 & 0 & 1 & 0 \\
1 & 1 & 1 & 1 & 0 & 0 & 0 & 0 & 0 & 1 \\
\end{array}\right) \in \mathbb F_2^{6\times 10}.$$

Algorithm \ref{Algorithm::1} returns the whole set of coset leaders of $\mathcal C$ described in Table \ref{Table::1} ordered w.r.t. the degree reverse lexicographic order $\prec$.
We denote by $\mathrm{CL}(\mathcal C)_{i}^j$ the $j$-th element of the set of coset leaders of weight $i$.

\begin{table}[!h]
\begin{center}
\begin{tabular}{|l|l|}
\hline
\multicolumn{2}{|c|}{Coset Leaders $\mathrm{CL}(\mathcal C)$} \\
\hline
$\mathrm{CL}(\mathcal C)_{0}$ & $[ \mathbf 0 ]$ \\ \hline
$\mathrm{CL}(\mathcal C)_{1}$ &  
$[ \mathbf e_1 ],~[ \mathbf e_2 ], ~[ \mathbf e_3 ], ~[ \mathbf e_4 ], ~[ \mathbf e_5 ],
~ [ \mathbf e_6 ], ~[ \mathbf e_7 ], ~[ \mathbf e_8 ], ~[ \mathbf e_9 ], ~[ \mathbf e_{10} ],$ \\ \hline
 
 \multirow{9}{*}{$\mathrm{CL}(\mathcal C)_{2}$} &  
$[\mathbf e_1+\mathbf e_2, \mathbf e_5+\mathbf e_6 ], 
~[\mathbf e_1+\mathbf e_3, \mathbf e_5+\mathbf e_7 ],
~[\mathbf e_1+\mathbf e_4, \mathbf e_5+\mathbf e_8 ],$\\
& $[\mathbf e_1+\mathbf e_5, \mathbf e_2+\mathbf e_6, \mathbf e_3+\mathbf e_7, \mathbf e_4+\mathbf e_8 ], 
~ [\mathbf e_1+\mathbf e_6, \mathbf e_2+\mathbf e_5],$\\
& $[\mathbf e_1+\mathbf e_7, \mathbf e_3+\mathbf e_5],
~ [\mathbf e_1+\mathbf e_8, \mathbf e_4+\mathbf e_5],
~ [\mathbf e_1+\mathbf e_9],
~ [\mathbf e_1+\mathbf e_{10}],$\\
& $ [\mathbf e_2+\mathbf e_3, \mathbf e_6+\mathbf e_7 ],
~ [\mathbf e_2+\mathbf e_4, \mathbf e_6+\mathbf e_8 ],
~ [\mathbf e_2+\mathbf e_7, \mathbf e_3+\mathbf e_6 ],$\\
& $[\mathbf e_2+\mathbf e_8, \mathbf e_4+\mathbf e_6 ],
~ [\mathbf e_2+\mathbf e_9 ],
~ [\mathbf e_2+ \mathbf e_{10}],$\\
& $[\mathbf e_3+\mathbf e_4, \mathbf e_7+\mathbf e_8 ],
~ [\mathbf e_3+\mathbf e_8, \mathbf e_4+\mathbf e_7 ],
~ [\mathbf e_3+\mathbf e_9 ],$\\ 
& $[\mathbf e_3+\mathbf e_{10}],
~ [\mathbf e_4+\mathbf e_9 ],
~ [\mathbf e_4+\mathbf e_{10}],
~ [\mathbf e_5+\mathbf e_9],$\\
& $[\mathbf e_5+\mathbf e_{10}],
~ [\mathbf e_6+\mathbf e_9],
~ [\mathbf e_6+\mathbf e_{10}],
~ [\mathbf e_7+\mathbf e_9],$\\
& $[\mathbf e_7+\mathbf e_{10}],
~ [\mathbf e_8+\mathbf e_9],
~ [\mathbf e_8+\mathbf e_{10}],
~ [\mathbf e_9+\mathbf e_{10}],$\\
\hline

 \multirow{17}{*}{$\mathrm{CL}(\mathcal C)_{3}$} &
$ [\mathbf e_1+\mathbf e_2+\mathbf e_3, \mathbf e_1+\mathbf e_6+\mathbf e_7, \mathbf e_2+\mathbf e_5+\mathbf e_7, \mathbf e_3+\mathbf e_5+\mathbf e_6 ],$\\
& $[\mathbf e_1+\mathbf e_2+\mathbf e_4, \mathbf e_1+\mathbf e_6+\mathbf e_8, \mathbf e_2+\mathbf e_5+\mathbf e_8, \mathbf e_4+\mathbf e_5+\mathbf e_6 ],$\\
& $[\mathbf e_1+\mathbf e_2+\mathbf e_7, \mathbf e_1+\mathbf e_3+\mathbf e_6, \mathbf e_2+\mathbf e_3+\mathbf e_5, \mathbf e_5+\mathbf e_6+\mathbf e_7 ],$\\
& $[\mathbf e_1+\mathbf e_2+\mathbf e_8, \mathbf e_1+\mathbf e_4+\mathbf e_6, \mathbf e_2+\mathbf e_4+\mathbf e_5, \mathbf e_5+\mathbf e_6+\mathbf e_8 ],$\\
& $[\mathbf e_1+\mathbf e_2+\mathbf e_9, \mathbf e_5+\mathbf e_6+\mathbf e_9 ],
~ [\mathbf e_1+\mathbf e_2+\mathbf e_{10}, \mathbf e_5+\mathbf e_6+\mathbf e_{10} ],$\\
& $[\mathbf e_1+\mathbf e_3+\mathbf e_4, \mathbf e_1+\mathbf e_7+\mathbf e_8, \mathbf e_3+\mathbf e_5+\mathbf e_8, \mathbf e_4+\mathbf e_5+\mathbf e_7 ],$\\
& $[\mathbf e_1+\mathbf e_3+\mathbf e_8, \mathbf e_1+\mathbf e_4+\mathbf e_7, \mathbf e_3+\mathbf e_4+\mathbf e_5, \mathbf e_5+\mathbf e_7+\mathbf e_8],$\\
& $[\mathbf e_1+\mathbf e_3+\mathbf e_9, \mathbf e_5+\mathbf e_7+\mathbf e_9 ],
~ [\mathbf e_1+\mathbf e_3+\mathbf e_{10}, \mathbf e_5+\mathbf e_7+\mathbf e_{10} ],$\\
& $[\mathbf e_1+\mathbf e_4+\mathbf e_9, \mathbf e_5+\mathbf e_8+\mathbf e_9],
~ [\mathbf e_1+\mathbf e_4+\mathbf e_{10}, \mathbf e_5+\mathbf e_8+\mathbf e_{10} ],$\\
& $[\mathbf e_1+\mathbf e_5+\mathbf e_9, \mathbf e_2+\mathbf e_6+\mathbf e_9, \mathbf e_3+\mathbf e_7+\mathbf e_9, \mathbf e_4+\mathbf e_8+\mathbf e_9 ],$\\
& $[\mathbf e_1+ \mathbf e_5+\mathbf e_{10}, \mathbf e_2+\mathbf e_6+\mathbf e_{10}, \mathbf e_3+\mathbf e_7+\mathbf e_{10}, \mathbf e_4+\mathbf e_8+\mathbf e_{10} ],$\\
& $[\mathbf e_1+\mathbf e_6+\mathbf e_9, \mathbf e_2+\mathbf e_5+\mathbf e_9 ],
~ [\mathbf e_1+\mathbf e_6+\mathbf e_{10}, \mathbf e_2+\mathbf e_5+\mathbf e_{10}],$\\
& $[\mathbf e_1+\mathbf e_7+\mathbf e_9, \mathbf e_3+\mathbf e_5+\mathbf e_9],
~ [\mathbf e_1+\mathbf e_7+\mathbf e_{10}, \mathbf e_3+\mathbf e_5+\mathbf e_{10}],$\\
& $[\mathbf e_1+\mathbf e_8+\mathbf e_9, \mathbf e_4+\mathbf e_5+\mathbf e_9], 
~ [\mathbf e_1+\mathbf e_8+\mathbf e_{10}, \mathbf e_4+\mathbf e_5+\mathbf e_{10}],$\\
& $[\mathbf e_1+\mathbf e_9+\mathbf e_{10}],$\\
& $[\mathbf e_2+\mathbf e_3+\mathbf e_8, \mathbf e_2+\mathbf e_4+\mathbf e_7, \mathbf e_3+\mathbf e_4+\mathbf e_6, \mathbf e_6+\mathbf e_7+\mathbf e_8 ],$\\
& $[\mathbf e_5+\mathbf e_9+\mathbf e_{10}]$\\
\hline
\end{tabular}
\caption{Set of coset-leaders of Example \ref{Example::1}}
\label{Table::1}
\end{center}
\end{table}

The main difference between this paper and previous works is the considera\-tion of all coset leaders and not just those belonging to $\mathcal N$. Note that no subword of two elements of $\mathbf y = \mathbf e_4 + \mathbf e_5 + \mathbf e_6 \in \mathrm{CL}(\mathcal C)_3^2$ is part of $\mathcal N$,  i.e. 
$\mathbf e_4 + \mathbf e_5 \in \mathrm{CL}(\mathcal C)_2^7$,
$\mathbf e_4 + \mathbf e_6 \in \mathrm{CL}(\mathcal C)_2^{13}$ 
and $\mathbf e_5 + \mathbf e_6 \in \mathrm{CL}(\mathcal C)_2^1$ do not lie in $\mathcal N$.
Therefore the importance of the second property of the Definition \ref{List-Definition} to obtain the complete set of coset leaders.

Algorithm \ref{Algorithm::1} could be adapted without incrementing the complexity to get more information such as:
\begin{itemize}
\item The \emph{Newton radius} $\nu(\mathcal C)$ of a binary code $\mathcal C$ is the largest weight of any error vector that can be uniquely corrected, or equivalently, $\nu (\mathcal C)$ is the largest value among the cosets with only one coset leader.
In our example it suffice to analyze the last element of the list $\mathrm{CL}(\mathcal C)$ to obtain the coset of highest weight which contains only one leader, i.e. 
$\nu(\mathcal C) = 3$ since $\mathrm{CL}(\mathcal C)_{3}^{23} = [\mathbf e_5 + \mathbf e_9 + \mathbf e_{10}]$.

\item The \emph{covering radius} $\rho(\mathcal C)$ of a binary code $\mathcal C$ is the smallest integer $s$ such that $\mathbb F_2^n$ is the union of the spheres of radius $s$ centered at the codewords of $\mathcal C$, i.e. 
$\rho (\mathcal C) = \max_{\mathbf y \in \mathbb F_2^n} \min_{\mathbf c \in \mathcal C} d_H(\mathbf y, \mathbf c)$. 
It is well known that $\rho(\mathcal C)$ is the weight of the coset of largest weight.
Likewise, in our example $\rho(\mathcal C) = 3$ since $\mathrm{CL}(\mathcal C)_{3}^{23} = [\mathbf e_5 + \mathbf e_9 + \mathbf e_{10}]$ is the coset of highest weight.

\item The \emph{Weight Distribution of the Coset Leaders} of a binary code $\mathcal C$ is the list 
$\mathrm{WDCL} = (\alpha_0, \ldots, \alpha_n)$ where $\alpha_i$ with $1\leq i \leq n$ is the number of cosets with coset leaders of weight $i$. Note that the set $\mathcal N$ is enough to compute this parameter. It is clear that 
$$\mathrm{WDCL} = \left[\begin{array}{cccccccccc}
1, & 10,& 30, & 23,& 0, & 0,& 0,& 0,& 0,& 0
\end{array}\right].$$

\item The number of coset leaders in each coset :
$$\sharp \left(\mathrm{CL} \right) = \left[\begin{array}{l}
1, \\
1, 1, 1, 1, 1, 1, 1, 1, 1, 1, \\
2, 2, 2, 4, 2, 2, 2, 1, 1, 2,
2, 2, 2, 1, 1, 2, 2, 1, 1, 1,
1, 1, 1, 1, 1, 1, 1, 1, 1, 1, \\
4, 4, 4, 4, 2, 2, 4, 4, 2, 2, 
2, 2, 4, 4, 2, 2, 2, 2, 2, 2,
1, 4, 1\\
\end{array}\right]$$
Note that there are $30$ of the $64$ cosets where the Complete Decoding Problem (CDP) has a unique solution. It is also interesting to note that among the cosets with one leaders there are more cosets exceeding the error correction capacity ($19$) than achieving such capacity ($11$).
\end{itemize}
\end{example}

\subsection{Complexity Analysis}
The next theorem states an upper bound for the number of iterations that Algorithm \ref{Algorithm::1} will perform.

\begin{theorem}
Algorithm \ref{Algorithm::1} computes the set of coset leaders of a given binary code $\mathcal C$
of length $\mathbf n$ after at most $\mathbf n|\mathrm{CL}(\mathcal C)|$ iterations.
\end{theorem}
\begin{proof}
Let $\tt{List}$ be the set constructed in the proof of Theorem \ref{Theorem3}.
Notice that by looking how Algorithm \ref{Algorithm::1} is constructed, the number of iterations is exactly the size of ${\tt List}$. Moreover note that we can
write $\tt{List}$ as the following set
$${\tt List}=\{\mathbf w + \mathbf e_i\mid \mathbf w\in\mathrm{CL}(\mathcal
C)\mbox{ and }{i \in \left\lbrace 1,\ldots ,n\right\rbrace \}}.$$
Therefore it is clear that the size of $\tt{List}$ is bounded by $\mathbf n|\mathrm{CL}(\mathcal C)|$.
\end{proof}

\begin{remark}
\begin{enumerate}
\item We can proceed analogously to the previous proof to estimate the required memory space which is $\mathcal O\left(\mathbf n|\mathrm{CL}(\mathcal C)|\right)$. In the best case, $\mathcal O\left(|\mathrm{CL}(\mathcal C)|\right)$ of memory space is needed, thus Algorithm \ref{Algorithm::1} is near the optimal case when considering memory requirements. However the order of the set $\mathrm{CL}(\mathcal C)$ is exponential on the codimension of the code, i.e. $\mathcal O\left( 2^{n-k}\right)$, so this method is impractical for large codes.

\item Algorithm \ref{Algorithm::1} generates at most $\mathbf n|\mathrm{CL}(\mathcal C)|$ words from $\mathbb F_2^n$ to compute the set of all coset leaders. Therefore, the proposed algorithm has near-optimal performance and significantly reduced complexity.
\end{enumerate}
\end{remark}

Note that the statement $\mathrm B(\mathbf c, e) \cap \mathrm B(\hat{\mathbf c}, e)= \emptyset$ holds true for all $\mathbf c, \hat{\mathbf c}\in \mathcal C$ with $\mathbf c \neq \hat{\mathbf c}$ if and only if $2e+1 \leq d(\mathcal C)$ is valid. Moreover, $\mathbb F_2^n = \cup_{\mathbf c\in \mathcal C} \mathrm B(\mathbf c, e)$ holds true if and only if the covering radius satisfies that $\rho (\mathcal C) \leq e$. Therefore the minimum distance and the covering radius of any code are related by $d(\mathcal C)\leq 2 \rho(\mathcal C) +1$.

\begin{lemma}
For any $[n,k]$ binary code $\mathcal C$ the following inequality holds:
$$\sum_{i=0}^t \binom{n}{i} \leq |\mathrm{CL}(\mathcal C)| \leq \sum_{j=0}^{\rho(\mathcal C)} \binom{n}{j},$$
where $t$ denotes the error-correcting capacity of $\mathcal C$ and $\rho(\mathcal C)$ its covering radius.
\end{lemma}

\begin{proof}
Let us first prove that every vector $\mathbf e \in \mathbb F_2^n$ with $\mathrm w_H(\mathbf e)\leq t$ is a coset leader. Assume to the contrary that there exists a vector $\mathbf e \in \mathbb F_2^n$ with $\mathrm w_H(\mathbf e)\leq t$ and $\mathbf e \notin \mathrm{CL}(\mathcal C)$. Hence there is another vector $\hat{\mathbf e}\in \mathbb F_2^n$ with $S(\mathbf e) = S(\hat{\mathbf e})$ and $\mathrm w_H(\hat{\mathbf e})< \mathrm w_H(\mathbf e)$. Or equivalently, there exists a codeword $\mathbf e - \hat{\mathbf e}\in \mathcal C$ with $$\mathrm w_H(\mathbf e - \hat{\mathbf e}) \leq \mathrm w_H(\mathbf e) + \mathrm w_H(\hat{\mathbf e}) \leq 2t \leq d(\mathcal C)-1$$
which is a contradiction to the definition of the minimum distance of $\mathcal C$.
%For the previous inequalities, recall that if $\mathbf x, \mathbf y \in \mathbb F_2^n$ then 
%$$\mathrm w_H(\mathbf x + \mathbf y ) = \mathrm w_H(\mathbf x) + \mathrm w_H(\mathbf y) - 2 |\mathrm{supp}(\mathbf x)\cap \mathrm{supp}(\mathbf y)|.$$
%Moreover the error-correcting capacity of a linear code $\mathcal C$ is defined as $t= \left\lfloor \frac{d(\mathcal C)-1}{2}\right\rfloor$. 
Hence, we have actually proved that the number of vectors of weight up to $t$ is a lower bound for the cardinality of the set $\mathrm{CL}(\mathcal C)$, i.e.
$$\sum_{i=0}^t \binom{n}{i} \leq |\mathrm{CL}(\mathcal C)|.$$

Furthermore, by the definition of the covering radius of $\mathcal C$, we have that for all $\mathbf y \in \mathbb F_2^n$ there exists a codeword $\mathbf c \in \mathcal C$ such that $\mathrm d_H(\mathbf c, \mathbf y)\leq \rho(\mathcal C)$. In other words, there exists a vector $\mathbf e \in \mathbb F_2^n$ such that $\mathrm w_H(\mathbf e)\leq \rho(\mathcal C)$ and $S(\mathbf e) = S(\mathbf y)$. Thus, $\mathrm w_H\left(\mathrm{CL}(\mathbf y)\right)\leq \rho(\mathcal C)$ and the lemma holds. 
\end{proof}

If the above lemma holds with equality then $\mathcal C$ is called a \emph{perfect} code. That is to say, let $\mathcal C$ be a linear code with more than one codeword, then $\mathcal C$ is a perfect code if and only if $\rho(\mathcal C) = t$.

\section{Computing a test set}
\label{Section4}

In this section we show how Algorithm \ref{Algorithm::1} can be adapted to compute a \emph{test-set} for a binary linear code.

\begin{definition}
\label{LeaderCodewords}
The set of \emph{leader codewords} of a given binary code $\mathcal C$  is defined as:
$$\mathrm L(\mathcal C) = \left\{ 
\mathbf n_1 + \mathbf n_2 + \mathbf e_i \in \mathcal C\setminus \{\mathbf 0 \} \mid 
\begin{array}{ccc}
i \notin \mathrm{supp}(\mathbf n_1) 
&\hbox{ and }& 
\mathbf n_1, \mathbf n_2 \in \mathrm{CL}(\mathcal C)
\end{array}
\right\}$$
\end{definition}

For efficiency reasons we are just interested in a particular case of the above object, when 
$\mathrm{supp}(\mathbf n_1 + \mathbf e_i) \cap \mathrm{supp}(\mathbf n_2) = \emptyset$.

%%%%% ALGORITMO CLBC2
\begin{algorithm2e}[!h]
\KwData{A weight compatible ordering $\prec$ and a parity check matrix $H$ of a binary code $\mathcal C$.} 
\KwResult{The set of coset leaders $\mathrm{CL}(\mathcal C)$ and the set of leader codewords $\mathrm{L}(\mathcal C)$ for $\mathcal C$.}
${\tt Listing} \longleftarrow [0]$;
%$\mathcal N \longleftarrow \emptyset$;
$r \longleftarrow 0$;
$\mathrm{CL}(\mathcal C) \longleftarrow \emptyset$;
$\mathcal S \longleftarrow \emptyset$;
$\mathrm{L}(\mathcal C) \longleftarrow \emptyset$\;
\While{${\tt Listing} \neq \emptyset$}
{
	$\mathbf t \longleftarrow {\tt NextTerm}[{\tt Listing}]$\;
	$\mathbf s \longleftarrow \mathbf t H^T$\;
	$j \longleftarrow {\tt Member}[s,\mathcal S]$\;
	\eIf{$j \neq {\tt false}$}
	{
		\If{$\mathrm{w}_H(\mathbf t) = \mathrm{w}_H(\mathrm{CL}(\mathcal C)[j][1])$}
		{
		$\mathrm{CL}(\mathcal C)[j] \longleftarrow \mathrm{CL}(\mathcal C)[j] \cup \{ \mathbf t\}$\;
		${\tt Listing} \longleftarrow {\tt InsertNext}[\mathbf t, {\tt Listing}]$\;
		}
		\For{$i \in \mathrm{supp}(\mathbf t)~:~ \mathbf t = \mathbf t' + \mathbf e_i$ with $\mathbf t' \in \mathrm{CL}(\mathcal C)$ and $i \notin \mathrm{supp}(\mathbf t')$}
		{
		\hspace{-0.1cm}$\mathrm{L}(\mathcal C) \longleftarrow \mathrm L(\mathcal C) \cup \left\{\mathbf t + \mathbf t_k \mid \mathbf t_k \in \mathrm{CL}(\mathcal C)[j] {\small \hbox{ and }}\mathrm{supp}(\mathbf t)\cap \mathrm{supp}(\mathbf t_k) = \emptyset\right\}$
		}
	}
	{ 
		$r\longleftarrow r+1$;  
		%$\mathbf t_r \longleftarrow \mathbf t$\;
		%$\mathcal N \longleftarrow \mathcal N \cup \{ \mathbf t_r\}$\;
		$\mathrm{CL}(\mathcal C)[r] \longleftarrow \{ \mathbf t\}$;
		$\mathcal S[r] \longleftarrow \mathbf s$\;
		${\tt Listing} = {\tt InsertNext}[\mathbf t, {\tt Listing}]$\;
	}
} 
\caption{Computation of a \emph{test-set} for $\mathcal C$}
\label{Algorithm::2}
\end{algorithm2e}

%%%%%%%%% End Algorithm CLBC2

\begin{remark}
\label{Remark2}
The difference between Algorithm \ref{Algorithm::1} and Algorithm \ref{Algorithm::2} are \textbf{Steps 10-12} from Algorithm \ref{Algorithm::2}.
\end{remark}

\begin{theorem}
\label{Theorem4}
Algorithm \ref{Algorithm::2} computes the set of coset leaders and the set of leader codewords of a given binary code $\mathcal C$.
\end{theorem}
\begin{proof}
Taking into account Remark \ref{Remark2} and Theorem \ref{Theorem3} we only need to prove that Algorithm \ref{Algorithm::2} computes the set of leader codewords.

We first observe that all the words inserted in the set $\mathrm L(\mathcal C)$ during Algorithm \ref{Algorithm::2} are leader codewords. These elements are of the type $\mathbf v = \mathbf t'+\mathbf e_i + \mathbf t_k$ where $\mathbf t=\mathbf t'+\mathbf e_i$ and $\mathbf t_k$ are in the same coset. {Moreover we have that
$$\begin{array}{cccc}
\mathbf t', \mathbf t_k \in \mathrm{CL}(\mathcal C), & 
\mathrm{supp}(\mathbf t') \cap \mathrm{supp}(\mathbf t_k) = \emptyset & \hbox{ and } &
i\notin \mathrm{supp}(\mathbf t').
\end{array}$$}
Therefore, by Definition \ref{LeaderCodewords} $\mathbf v$ is a leader codeword.

Note that the list $\tt{Listing}$ is in ascending order w.r.t. $\prec$, {therefore in each loop we study all leader codewords of the form 
$\mathbf n_1 + \mathbf e_i + \mathbf n_2$ with $\mathbf n_1, \mathbf n_2 \leq \mathbf t$}.
The fact that Theorem \ref{Theorem3} shows that Algorithm \ref{Algorithm::2} computes the whole set $\mathrm{CL}(\mathcal C)$ proves that all leader codewords are introduced in $\mathrm L(\mathcal C)$.
\end{proof}

By its construction, Algorithm \ref{Algorithm::2} has the same time complexity as Algorithm \ref{Algorithm::1}. The advantage of computing the set of leader codewords is that it helps in solving the same problems as the function \emph{Matphi} does but with a structure which is considerately smaller.

\begin{definition}
We define the subset $\mathrm L^1(\mathcal C)$ of $\mathrm L(\mathcal C)$ as
$$\mathrm L^1(\mathcal C) = \left\{ 
\mathbf n_1 + \mathbf n_2 + \mathbf e_i \in \mathcal C \setminus \{ \mathbf 0\}
\left| 
\begin{array}{c}
i\notin \mathrm{supp}(\mathbf n_1), ~
\mathbf n_1 \in \mathrm{CL}(\mathcal C), ~
 \mathbf n_2 \in \mathcal N\\
 \hbox{and }
\mathrm{w}_H(\mathbf n_1 + \mathbf e_i) > \mathrm{w}_H(\mathbf n_2), \\
\end{array}
\right.\right\}.$$
\end{definition}

\begin{remark}
Note that the condition $\mathrm{w}_H(\mathbf n_1 + \mathbf e_i) > \mathrm{w}_H(\mathbf n_2)$ is imposed just to improve the efficiency of computing this set. Therefore $\mathrm L(\mathcal C)$ can be rewritten as
{$$\mathrm L(\mathcal C) = \left\{ 
\mathbf n_1 + \mathbf n_2 + \mathbf e_i \in \mathcal C\setminus \{\mathbf 0 \} \left|
\begin{array}{c}
i \notin \mathrm{supp}(\mathbf n_1), 
\mathbf n_1, \mathbf n_2 \in \mathrm{CL}(\mathcal C)\\
\hbox{and } \mathrm{w}_H(\mathbf n_1 + \mathbf e_i) > \mathrm{w}_H(\mathbf n_2)
\end{array}\right.
\right\}.$$}

Thus, the only difference between the sets $\mathrm L^1(\mathcal C)$ and $\mathrm L(\mathcal C)$ is that $\mathbf n_2 \in \mathcal N$ instead of $\mathbf n_2 \in \mathrm{CL}(\mathcal C)$. In other words, to the element $\mathbf n_2$ in $\mathrm L^1(\mathcal C)$ is required not only to belong to the set of the coset leaders but also to be the smallest element in its coset according to a fixed weight compatible ordering $\prec$.
\end{remark}

\begin{theorem}
\label{Theorem5}
The subset $\mathrm L^1(\mathcal C)$ of $\mathrm L(\mathcal C)$ is a test-set for $\mathcal C$.
\end{theorem}

\begin{proof}
Let us consider a word $\mathbf y \notin \mathrm{CL}(\mathcal C)$ with
$\mathrm{supp}(\mathbf y) = \left\{ i_1, \ldots, i_m\right\}\subseteq [1, n]$. Thus, there must exist an integer $1\leq l <m$ such that
$$\begin{array}{ccc}
\mathbf n_1:= \mathbf e_{i_1} + \ldots + \mathbf e_{i_l} \in \mathrm{CL}(\mathcal C) &
\hbox{ and }&
\mathbf n_1 + \mathbf e_{i_{l+1}}\notin \mathrm{CL}(\mathcal C).
\end{array}$$

We define $\mathbf n_2 = N\left(\mathbf n_1 + \mathbf e_{i_{l+1}}\right)$, i.e. $\mathbf n_2$ is the smallest element in the coset of $\mathbf n_1 + \mathbf e_{i_{l+1}}$ according to a fixed compatible weight ordering $\succ$. Since $\mathbf n_1 + \mathbf e_{i_{l+1}}\notin \mathrm{CL}(\mathcal C)$ we have that $\mathrm w_H(\mathbf n_2)< \mathrm w_H(\mathbf n_1 + \mathbf e_{i_{l+1}})$. Thus, $\mathbf t = \mathbf n_1 + \mathbf n_2 + \mathbf e_{i_{l+1}}\in \mathrm L^1(\mathcal C)$.

Without loss of generality we may assume that $\mathrm{supp}(\mathbf n_1 + \mathbf e_{i_{l+1}})\cap \mathrm{supp}(\mathbf n_2) = \emptyset$. Indeed,
\begin{itemize}
\item if $i_{l+1}\in \mathrm{supp}(\mathbf n_2)$ then, by Theorem \ref{Theorem1}, $\mathbf n_2 + \mathbf e_{i_{l+1}}\in \mathrm{CL}(\mathcal C)$. Moreover 
$$\begin{array}{ccc}
S(\mathbf n_2 + \mathbf e_{i_{l+1}}) = S(\mathbf n_1) & \hbox{ and }&
\mathrm w_H(\mathbf n_2 + \mathbf e_{i_{l+1}})<\mathrm w_H(\mathbf n_1)
\end{array},$$ which contradicts the fact that $\mathbf n_1 \in \mathrm{CL}(\mathcal C)$. 
\item Otherwise, if there exists $j\in \mathrm{supp}(\mathbf n_1)\cap \mathrm{supp}(\mathbf n_2)$. Then, we replace the elements $\mathbf n_1$ and $\mathbf n_2$ by
$\begin{array}{ccc}
\overline{\mathbf n_1} = \mathbf n_1 + \mathbf e_j  & \hbox{ and } & \overline{\mathbf n_2} = \mathbf n_2 + \mathbf e_j
\end{array}.$
Note that, by Theorem \ref{Theorem1}, $\overline{\mathbf n_1}, \overline{\mathbf n_2}\in \mathrm{CL}(\mathcal C)$. Furthermore, 
$$\begin{array}{ccc}
\mathrm w_H(\overline{\mathbf n_1} + \mathbf e_{i_{l+1}}) > \mathrm w_H(\overline{\mathbf n_2}) & \hbox{ and }&
S\left(\overline{\mathbf n_1} + \mathbf e_{i_{l+1}}\right) = S(\overline{\mathbf n_2}).
\end{array}$$ 
Thus we still have that $\mathbf t = \overline{\mathbf n_1} + \overline{\mathbf n_2} + \mathbf e_{i_{l+1}}\in \mathrm L^1 (\mathcal C)$.
\end{itemize}

Therefore,
$|\mathrm{supp}(\mathbf t)\cap \mathrm{supp}(\mathbf y)| \geq \mathrm{w}_H(\mathbf n_1 + \mathbf e_{i_{l+1}}) > \mathrm{w}_H(\mathbf n_2) 
\geq |\mathrm{supp}(\mathbf t)\cap \mathrm{supp}(\overline{\mathbf y})|$
where $\overline{\mathbf y}$ denotes the relative complement of $\mathbf y$ in $\mathbb F_2^n$, and in consequence,
$\mathrm{w}_H(\mathbf y -\mathbf t)<\mathrm{w}_H(\mathbf y)$ which completes the proof.
\end{proof}

Since $\mathrm L^1(\mathcal C)\subseteq \mathrm L(\mathcal C)$ and by Theorem \ref{Theorem5} the subset $\mathrm L^1(\mathcal C)$ forms a test-set for $\mathcal C$, then so does the set $\mathrm L(\mathcal C)$.
The following theorem gives a bound for the weight of a leader codeword of a given binary code $\mathcal C$. 

\begin{theorem}
\label{Theorem6}
Let $\mathbf c \in \mathrm L(\mathcal C)$ then $\mathrm{w}_H(\mathbf c) \leq 2 \rho(\mathcal C) +1$ where $\rho(\mathcal C)$ is the covering radius of $\mathcal C$.
\end{theorem}
\begin{proof}
Let $\mathbf c \in \mathrm L(\mathcal C)$ then there exists $\mathbf n_1,~\mathbf n_2 \in \mathrm{CL}(\mathcal C)$ and $i \notin \mathrm{supp}(\mathbf n_1)$ such that 
$\mathrm{w}_H(\mathbf n_1 + \mathbf e_i) > \mathrm{w}_H(\mathbf n_2)$ and $\mathbf c = \mathbf n_1 + \mathbf e_i + \mathbf n_2$. Applying the definition of covering radius we have that $\mathrm{w}_H(\mathbf n_1), \mathrm{w}_H(\mathbf n_2) \leq \rho$, thus $\mathrm{w}_H(\mathbf c) \leq 2\rho +1$.
\end{proof}

In Algorithm \ref{Algorithm::3} we describe a method to compute the subset $\mathrm{CL}(\mathbf y)$ of coset leaders corresponding to the coset $\mathcal C + \mathbf y$. 
Note that we first need to achieve the element $N(\mathbf y)$. We propose to use a Gradient Descent Decoding Algorithm (GDDA) for this purpose. This approach resembles those techniques presented in \cite{borges:2011}.

%%%%% ALGORITMO3
\begin{algorithm2e}[!h]
\KwData{A received vector $\mathbf y \in \mathbb F_2^n$ and the set of leader codewords $\mathrm L(\mathcal C)$ of a binary code $\mathcal C$.} 
\KwResult{The subset $\mathrm{CL}(\mathbf y)$ of coset leaders  corresponding to the coset $\mathcal C + \mathbf y$.}

\Begin({\textbf{: Computes $N(\mathbf y)$ by a GDDA\footnotemark[1]using $\mathrm L(\mathcal C)$ as a test-set for $\mathcal C$}\label{GDDA}}){
$N(\mathbf y) \longleftarrow \mathbf 0$\;
\While{there exist $\mathbf t \in \mathrm L(\mathcal C)$ such that $\mathrm w_H(\mathbf y + \mathbf t) < \mathbf w_H(\mathbf y)$}
{
	$\mathbf c \longleftarrow \mathbf c + \mathbf t$\;
	$\mathbf y \longleftarrow \mathbf y + \mathbf t$\;
}
}
$\mathbf y \longleftarrow N(\mathbf y)$; 
$\mathcal S \longleftarrow \{ \mathbf y \}$;
$L \longleftarrow \mathrm L(\mathcal C)$\;
\While{there exists $\mathbf c \in L~:~ \mathrm w_H(\mathbf y - \mathbf c) = \mathrm w_H(\mathbf y)$}
{
	$\mathbf y \longleftarrow \mathbf y - \mathbf c$;
	$\mathcal S \longleftarrow \mathcal S \cup \{ \mathbf y \}$\;
	$L \longleftarrow L - \{ \mathbf c\}$\;
} 
\caption{Computing the set $\mathrm{CL}(\mathbf y)$}
\label{Algorithm::3}
\end{algorithm2e}
%%%%%%%%% End Algorithm::3

\begin{theorem}
\label{Theorem7}
Algorithm \ref{Algorithm::3} computes, from the set $\mathrm L(\mathcal C)$, the subset $\mathrm{CL}(\mathbf y)$ of coset leaders  corresponding to the coset $\mathcal C + \mathbf y$ for a given received vector $\mathbf y \in \mathbb F_2^n$.
\end{theorem}

\begin{proof}
Let us first prove that every $\mathbf z \in \mathrm{CL}(\mathbf y)$ can be rewritten as $\mathbf z = N(\mathbf y) - \mathbf c$ with $\mathbf c \in \mathrm L(\mathcal C)$. Let $i \in \mathrm{supp}(\mathbf z)$ then $\mathbf z = \mathbf n_1 + \mathbf e_i$ with $i\notin \mathrm{supp}(\mathbf n_1)$. Hence, by Theorem \ref{Theorem1}, $\mathbf n_1 \in \mathrm{CL}(\mathcal C)$. 
Furthermore we have that
$$\begin{array}{ccc}
S(N(\mathbf y)) = S(\mathbf z) & \hbox{ and } & 
\mathrm w_H(\mathbf n_1) < \mathrm w_H(\mathbf z) = \mathrm w_H(N(\mathbf y)).
\end{array}$$
Thus, from the definition of leader codewords, $\mathbf c = N(\mathbf y) + (\mathbf n_1 + \mathbf e_i)\in \mathrm L(\mathcal C)\subseteq \mathcal C$, or equivalently, $\mathbf z = \mathbf n_1 + \mathbf e_i = \mathbf c - N(\mathbf y)$ with $\mathbf y \in L(\mathcal C)$.

The proof is completed by noting that Theorem \ref{Theorem5} guarantees \textbf{Step 1}.
\end{proof}

\subsection{Leader codewords and zero neighbours}
\footnotetext[1]{GDDA is the abbreviation for Gradient Descent Decoding Algorithm}
In this section we will give a brief review of basic concepts from \cite[Section 3]{barg:1998} and thus establish the relation between zero neighbours and leader codewords of a binary code $\mathcal C$.

\begin{definition}
For any subset $A \subset \mathbb F_2^n$ we define $\mathcal X(A)$ as the set of words at Hamming distance $1$ from $A$, i.e.
$$\mathcal X (A) = \left\{ \mathbf y \in \mathbb F_2^n \mid \min \left\{ d_H(\mathbf y, \mathbf a)~:~ \mathbf a \in A\right\} = 1\right\}.$$
We define the \emph{boundary} of $A$ as $\delta (A) = \mathcal X (A) \cup \mathcal X (\mathbb F_2^n \setminus A)$.
\end{definition}

\begin{definition}
A nonzero codeword $\mathbf c \in \mathcal C$ is called \emph{zero neighbour} if its Voronoi region shares a common boundary with the set of coset leaders, i.e. 
$$\delta (\mathrm D(\mathbf z)) \cap \delta (\mathrm D(\mathbf 0)) \neq \emptyset.$$
We will denote by $\mathcal Z(\mathcal C)$ the set of all zero neighbours of $\mathcal C$ that is to say:
{$$\mathcal Z(\mathcal C) = \left\{ \mathbf z \in \mathcal C \setminus \{ \mathbf 0\}
~:~ \delta (\mathrm D(\mathbf z)) \cap \delta (\mathrm D(\mathbf 0)) \neq \emptyset
\right\}.$$}
\end{definition}

Note that if $\mathbf z \in \mathcal C \setminus \{ \mathbf 0 \}$ satisfies that 
$\mathcal X (\mathrm D (\mathbf 0)) \cap \mathrm D(\mathbf z) \neq \emptyset$, then $\mathbf z \in \mathcal Z(\mathcal C)$. Furthermore $\mathcal Z(\mathcal C)$ is a test-set for $\mathcal C$ (see for instance \cite[Theorem 3.16]{barg:1998}). However the only property of the set $\mathcal Z(\mathcal C)$ that is essential for decoding is
$$\mathcal X(\mathrm D(\mathbf 0)) \subseteq \bigcup_{\mathbf z \in \mathcal Z(\mathcal C)} \mathrm D(\mathbf z).$$

Thus, if we restrict the set $\mathcal Z(\mathcal C)$ to a smallest subset verifying the previous property we still have a test-set for $\mathcal C$. We will denote such subset of $\mathcal Z(\mathcal C)$ by $\mathcal Z_{\min}(\mathcal C)$. Note that the set $\mathcal Z_{\min}(\mathcal C)$ may be not unique, however its size is well defined.

\begin{theorem}
\label{Theorem8}
Let $\mathcal C$ be a binary code and $\mathbf z \in \mathcal C \setminus \{\mathbf 0\}$. Then the following are equivalent:
\begin{enumerate}
\item $\mathcal X(\mathrm D(\mathbf 0)) \cap \mathrm D(\mathbf z) \neq \emptyset$.
\item $\mathbf z \in \mathrm L(\mathcal C)$.
\end{enumerate}
\end{theorem}
\begin{proof}
If $\mathcal X(\mathrm D(\mathbf 0)) \cap \mathrm D(\mathbf z) \neq \emptyset$ then there exists $\mathbf n_1 \in \mathrm D(\mathbf 0) = \mathrm{CL}(\mathcal C)$ and $i \notin \mathrm{supp}(\mathbf n_1)$ such that $\mathbf n_1 + \mathbf e_i \in \mathcal X(\mathrm D(\mathbf 0))$ and $\mathbf n_1 + \mathbf e_i \in \mathrm D(\mathbf z)$. In other words,
\begin{equation}
\label{Equation1}
\mathrm w_H(\mathbf z - (\mathbf n_1 + \mathbf e_i)) \leq \mathrm w_H(\mathbf c - (\mathbf n_1 + \mathbf e_i)) \hbox{ for all } \mathbf c \in \mathcal C\setminus \{ \mathbf z\},
\end{equation}
or equivalently, $\mathbf n_2 = \mathbf z - (\mathbf n_1 + \mathbf e_i) \in \mathrm{CL}(\mathcal C)$ with $\mathbf z \in \mathcal C$, thus 
$S(\mathbf n_2) = S(\mathbf n_1 + \mathbf e_i)$.
Furthermore, the special case of $\mathbf c = \mathbf 0 \in \mathcal C \setminus \{ \mathbf z \}$ of Equation \ref{Equation1} implies that $\mathrm w_H(\mathbf n_2)\leq \mathrm w_H(\mathbf n_1 + \mathbf e_i)$.
Therefore, all conditions in Definition \ref{LeaderCodewords} are verified, i.e. $\mathbf z = \mathbf n_1 + \mathbf n_2 + \mathbf e_i \in \mathrm L(\mathcal C)$.

Conversely, if $\mathbf z \in \mathrm L(\mathcal C)$, then $\mathbf z$ can be rewritten as $\mathbf z = \mathbf n_1 + \mathbf n_2 + \mathbf e_i$ where
$$\begin{array}{llccll}
(1) & \mathbf n_1, ~\mathbf n_2 \in \mathrm{CL}(\mathcal C). & ~~~ & ~~~ &
(3) & \mathrm w_H(\mathbf n_1 + \mathbf e_i) > \mathrm w_H(\mathbf n_2).\\
(2) & i \notin \mathrm{supp}(\mathbf n_1). & & &
(4) & S(\mathbf n_2) = S(\mathbf n_1 + \mathbf e_i).
\end{array}$$

Now $(1)$ and $(2)$ gives that $\mathbf n_1 + \mathbf e_i \in \mathcal X(\mathrm D(\mathbf 0))$, whereas $(1)$, $(3)$ and $(4)$ clearly force that $\mathrm{CL}(\mathbf n_1 + \mathbf e_i) = \mathbf n_2$, i.e. $\mathrm w_H(\mathbf n_2) \leq \mathrm w_H(\mathbf n_1 + \mathbf e_i + \mathbf c)$ for all $\mathbf c \in \mathcal C$, or equivalently, $\mathbf n_1 + \mathbf e_i \in \mathrm D(\mathbf z)$. Therefore, 
$\mathcal X (\mathrm D(\mathbf 0)) \cap \mathrm D(\mathbf z) \neq \emptyset$.
\end{proof}

\begin{corollary}
\label{Corollary1}
Let $\mathcal C$ be a binary code then $\mathcal Z_{\min}(\mathcal C) \subseteq \mathrm L(\mathcal C)$, for any minimal test-set $\mathcal Z_{\min}(\mathcal C)$ obtained from $\mathcal Z(\mathcal C)$.
\end{corollary}

\begin{proof}
Let $\mathcal Z_{\min}(\mathcal C)$ be a minimal test-set of $\mathcal C$ obtained from $\mathcal Z(\mathcal C)$, then every $\mathbf z\in \mathcal Z_{\min}(\mathcal C)$ satisfies that $\mathcal X(\mathrm D(\mathbf 0)) \cap \mathrm D(\mathbf z) \neq \emptyset$.
Thus, by Theorem \ref{Theorem8}, we obtained the required result.
\end{proof}

Algorithm \ref{Algorithm::2} gives the set of leader codewords $\mathrm L(\mathcal C)$ of a binary code $\mathcal C$. Furthermore, any minimal test-set $\mathcal Z_{\min}$ is a subset of $\mathrm L(\mathcal C)$. Thus, after performing redundancy elimination to $\mathrm L(\mathcal C)$, a minimal test-set $\mathcal Z_{\min}$ can also be obtained.

\begin{example}
We use the same code of Example \ref{Example::1}. Algorithm \ref{Algorithm::2} returns $\mathrm L(\mathcal C)$ and $\mathrm L^1(\mathcal C)$, in this case we obtained that both sets coincide. We describe below the set of leader codewords with $14$ elements of the given binary code $\mathcal C$.

$$\mathrm L(\mathcal C)= \mathrm L^1(\mathcal C) =\left\{
\begin{array}{c}
\begin{array}{cc}
\mathbf e_3+ \mathbf e_4+ \mathbf e_7+ \mathbf e_8, &
\mathbf e_2+ \mathbf e_4+ \mathbf e_6+ \mathbf e_8, \\
\mathbf e_2+ \mathbf e_3+ \mathbf e_6+ \mathbf e_7,&
\mathbf e_1+ \mathbf e_4+ \mathbf e_5+ \mathbf e_8,\\
\mathbf e_1+ \mathbf e_3+ \mathbf e_5+ \mathbf e_7,&
\mathbf e_1+ \mathbf e_2+ \mathbf e_5+ \mathbf e_6,\\
\mathbf e_4+ \mathbf e_6+ \mathbf e_7+ \mathbf e_9+ \mathbf e_{10},& 
\mathbf e_3+ \mathbf e_6+ \mathbf e_8+ \mathbf e_9+ \mathbf e_{10},\\
\mathbf e_2+ \mathbf e_7+ \mathbf e_8+ \mathbf e_9+ \mathbf e_{10},&
\mathbf e_2+ \mathbf e_3+ \mathbf e_4+ \mathbf e_9+ \mathbf e_{10}, \\
\end{array} \\
\begin{array}{c}
\mathbf e_1+ \mathbf e_5+ \mathbf e_6+ \mathbf e_7+ \mathbf e_8+ \mathbf e_9+ \mathbf e_{10},\\
\mathbf e_1+ \mathbf e_3+ \mathbf e_4+ \mathbf e_5+ \mathbf e_6+ \mathbf e_9+ \mathbf e_{10},\\
\mathbf e_1+ \mathbf e_2+ \mathbf e_4+ \mathbf e_5+ \mathbf e_7+ \mathbf e_9+ \mathbf e_{10},\\
\mathbf e_1+\mathbf e_2+ \mathbf e_3+ \mathbf e_5+ \mathbf e_8+ \mathbf e_9+ \mathbf e_{10}
\end{array}
\end{array}\right\}.$$

Note that the only nonzero codeword of $\mathcal C$ that is missing in $\mathrm L(\mathcal C)$ is the codeword $\mathbf y = (1, 1, 1, 1, 1, 1, 1, 1, 0, 0)$ of weight $8$. This result is consistent with the fact that the covering radius of $\mathcal C$ is $\rho(\mathcal C) =3$, as shown in Example \ref{Example::1}, and the statement of Theorem \ref{Theorem6}, where we proved that the weight of a leader codeword is always less or equal to $2\rho(\mathcal C) +1 =7$. 
$$\mathcal C = \left\{\begin{array}{c}
(0, 0, 0, 0, 0, 0, 0, 0, 0, 0), ~
(1, 0, 0, 0, 1, 1, 1, 1, 1, 1), ~
(0, 1, 0, 0, 0, 0, 1, 1, 1, 1), \\
(1, 1, 0, 0, 1, 1, 0, 0, 0, 0), ~
(0, 0, 1, 0, 0, 1, 0, 1, 1, 1), ~
(1, 0, 1, 0, 1, 0, 1, 0, 0, 0), \\
(0, 1, 1, 0, 0, 1, 1, 0, 0, 0), ~
(1, 1, 1, 0, 1, 0, 0, 1, 1, 1), ~
(0, 0, 0, 1, 0, 1, 1, 0, 1, 1), \\
(1, 0, 0, 1, 1, 0, 0, 1, 0, 0), ~
(0, 1, 0, 1, 0, 1, 0, 1, 0, 0), ~
(1, 1, 0, 1, 1, 0, 1, 0, 1, 1), \\
(0, 0, 1, 1, 0, 0, 1, 1, 0, 0), ~
(1, 0, 1, 1, 1, 1, 0, 0, 1, 1), ~
(0, 1, 1, 1, 0, 0, 0, 0, 1, 1), \\
\textbf{(1, 1, 1, 1, 1, 1, 1, 1, 0, 0)}
\end{array}\right\}$$
\end{example}

In the following table we present the computation results of a binary Golay code and a binary BCH code.

\begin{table}[h!]
\begin{tabular}{|l|c|c|}
\hline
& $[23,12]$ Golay code & $[21, 12]$ BCH code \\
\hline
\textbf{Codewords ($2^k$)} & $4096$ & $4096$ \\
\hline
\textbf{Cosets ($2^{n-k}$)} & $2048$ & $512$ \\
\hline
\textbf{Leader codewords ($|\mathrm L(\mathcal C)|$)} & $253$ & $549$\\
\hline
{$|\mathrm L^1(\mathcal C)|$} & $253$ & $470$\\
\hline
\end{tabular}
\caption{Number of codewords, number of cosets, number of leader codewords and the cardinality of {$|\mathrm L^1(\mathcal C)|$} of the $[23,12,7]$ binary Golay code and the $[21,12,5]$ binary BCH code.}
\end{table}

Therefore, we show an example where the subsets $\mathrm L(\mathcal C)$ and $\mathrm L^1(\mathcal C)$ agree (this is not a surprise since the Golay code is a perfect code) and an example where the set $\mathrm L^1(\mathcal C)$ is smaller than $\mathrm L(\mathcal C)$. Also, note that both codes have the same number of codewords but the Golay code has four times the number of cosets of the BCH code. On the other hand, the number of leader codewords is less in the Golay code. 

\begin{lemma}
If $\mathcal C$ is a perfect code, then $|\mathrm L(\mathcal C)| = |\mathrm L^1(\mathcal C)|$.
\end{lemma}

\begin{proof}
If $\mathcal C$ is a perfect code then every coset of $\mathcal C$ has a unique coset leader. That is, $\mathcal N = \mathrm{CL}(\mathcal C)$. Recall that the only difference between the sets $\mathrm L^1(\mathcal C)$ and $\mathrm L(\mathcal C)$ is that the component $\mathbf n_2$ of any element $\mathbf b = \mathbf n_1 + \mathbf n_2 + \mathbf e_i$ from $\mathrm L^1(\mathcal C)$ is required to belong to $\mathcal N \subseteq \mathrm{CL}(\mathcal C)$. But in this case this difference doesn't exists.
\end{proof}

\section{Implementations}
\label{s:con}

All the algorithms of this paper have been implemented and added to the collection of programs and procedures \textbf{GBLA\_LC} (\emph{Gr\"obner Basis by Linear Algebra and Linear Codes}).
This framework consist of various files written in the GAP \cite{GAP4} language and included in GAP's package GUAVA 3.10. Also during the Google Summer of code of $2013$ (\href{http://www.google-melange.com/gsoc/homepage/google/gsoc2013}{http://www.google-melange.com/gsoc/homepage/google/gsoc2013})  the student Ver\'onica Suaste (CIMAT, M\'exico) implemented Algorithm \ref{Algorithm::1} for inclusion in Sage \cite{sage}.  The code is published at \href{http://trac.sagemath.org/ticket/14973}{http://trac.sagemath.org/ticket/14973} and it will be included in next releases of Sage.

\section*{Aknowledgements} The authors gratefully acknowledge the helpful comments and suggestions of the editor  and the anonymous referees which contribute to a considerable improvement of this work.

\bibliographystyle{plain}      % mathematics and physical sciences
\bibliography{CosetLeaders}

\end{document}